\documentclass[journal]{IEEEtran}
\usepackage{cite}
\usepackage{amsmath,amssymb,amsfonts}
\usepackage{algorithmic}
\usepackage{graphicx}
\usepackage{textcomp}
\usepackage[T1]{fontenc}
\usepackage[utf8]{inputenc}
\usepackage[font=small,labelfont=bf,tableposition=top]{caption}
\usepackage{verbatim}
\usepackage{tikz}

\usetikzlibrary{positioning,fit,arrows.meta,backgrounds}
\usetikzlibrary{automata}
\usepackage{amsmath}
\usepackage{amssymb}
\usepackage{amsfonts}
\usepackage{amsthm}
\usepackage[english]{babel}
\usepackage{graphicx}
\usepackage[caption=false]{subfig}
\usepackage{algorithm2e}
\usepackage{multirow}
\usepackage{pgf,tikz,pgfplots}
\usepackage{filecontents}
\pgfplotsset{compat=1.14}
\usetikzlibrary{shapes,arrows}
\usepackage{mathrsfs}
\usetikzlibrary{arrows.meta}
\usepackage{booktabs}
\usepackage{colortbl}

\newcommand*{\defeq}{\stackrel{\text{def}}{=}}
\newcommand{\Agents}{\mathcal{U}}
\newcommand{\Time}{\mathbb{T}}
\newcommand{\Actions}{\Sigma} 
\newcommand{\Elabel}{\mathcal{Y}}
\newcommand{\Vlabel}{\mathcal{X}}
\newcommand{\Bool}{\mathbb{B}}
\newcommand{\Safety}{\varphi}
\newcommand{\Lang}{\mathcal{L}}

\newcommand{\Replace}{\mathcal{R}}

\newcommand\dhxrightarrow[2][]{%
  \mathrel{\ooalign{$\xrightarrow[#1\mkern4mu]{#2\mkern4mu}$\cr%
  \hidewidth$\rightarrow\mkern4mu$}}
}

\newcommand{\Dj}[1]{\textcolor{red}{#1}}

\newcommand\Square[1]{+(-#1,-#1) rectangle +(#1,#1)}

\newtheorem{corollary}{Corollary}

\newtheorem{lemma}{Lemma}
\newtheorem{prop}{Proposition}
\newtheorem{theorem}{Theorem}

\newtheorem{example}{Example}
\newtheorem{prob}{Problem}
\newtheorem{assume}{Assumption}

\begin{document}
\title{Online Synthesis for Runtime Enforcement of Safety in Multi-Agent Systems}
\author{Dhananjay Raju,  Suda Bharadwaj, Ufuk Topcu, and Franck Djeumou 
%\thanks{This paragraph of the first footnote will contain the date on 
%which you submitted your paper for review. It will also contain support 
%information, including sponsor and financial support acknowledgment. For 
%example, ``This work was supported in part by the U.S. Department of 
%Commerce under Grant BS123456.'' }
\thanks{The authors are with The University of Texas at Austin, Austin, TX 78712 USA (e-mail: draju, suda.b, utopcu, and fdjeumou@utexas.edu). }} 
%example, F. A. Author is with the National Institute of Standards and 
%Technology, Boulder, CO 80305 USA (e-mail: author@boulder.nist.gov). }
%\thanks{S. B. Author, Jr., was with Rice University, Houston, TX 77005 USA. He is 
%now with the Department of Physics, Colorado State University, Fort %Collins, 
%CO 80523 USA (e-mail: author@lamar.colostate.edu).}
%\thanks{T. C. Author is with 
%the Electrical Engineering Department, University of Colorado, Boulder, CO 
%80309 USA, on leave from the National Research Institute for Metals, 
%Tsukuba, Japan (e-mail: author@nrim.go.jp).}}

\maketitle

\begin{abstract}
We study the problem of enforcing safety in multi-agent systems at runtime by modifying the system behavior if a potential safety violation is detected. Traditional runtime enforcement methods that solve a reactive synthesis problem at design time have two significant drawbacks.  Firstly, these techniques do not scale as one has to take into account all possible behaviors from every agent, and this is computationally prohibitive. Second, these approaches require every agent to know the state of every other agent.  We address these limitations through a new approach where online modifications to behavior are synthesized onboard every agent. There is an \emph{enforcer} onboard every agent, which can modify the behavior of only the corresponding agent. In this approach, which is naturally decentralized, the enforcer on every agent has two components: a pathfinder that corrects the behavior of the agent and an ordering mechanism that dynamically modifies the priority of the agent. The current priority of an agent determines if the enforcer uses the pathfinder to modify the behavior of the agent. 
We derive an upper bound on the maximum deviation for any agent from its original behavior, that is all agents make progress. We prove that the worst-case synthesis time is quadratic in the number of agents at runtime as opposed to exponential at design-time for the existing methods that rely on design-time computation merely. Additionally, we prove the completeness of the technique under some mild assumptions; that is, if the agents can progress safely, then enforcers will find this behavior.
We test the technique in collision avoidance scenarios. For 50 agents in a 50$\times$50 grid modeling the common workspace for the agents, the online synthesis requires only a few seconds per agent whenever a potential collision is detected. In contrast, the centralized design time synthesis of shields for a similar setting is intractable beyond four agents in a 5$\times$5 grid.
\end{abstract}


\section{Introduction}

%\Dj{Problem setting - multi-agent safety with limited communication}

Ensuring the safety of multi-agent systems is a crucial and challenging problem. We study this problem in a setting in which (i) the agents do not know the state of the other agents, (ii) the agents can only communicate if they are in a communication group, which depends on spatial proximity, and (iii) each agent can share only a limited amount of information with the other agents. 

Runtime enforcement is one approach for ensuring safety for multi-agent systems \cite{FalconeFM12}. Enforcers typically monitor the behavior of the system and modify the behavior, if they detect a potential unsafety. Shielding is an approach to runtime enforcement \cite{BloemKKW15,KonighoferABHKT17}. A shield is typically assumed to be aware of and be able to affect all the agents in the system instantaneously \cite{multiagentshield}. Thus, shields require global information about the state of the system. However,
global information on the state of all the agents is often difficult to obtain in multi-agent systems. There has been some work in relaxing these assumptions using \emph{localized shields} that have awareness and authority over only the agents in their local region. However, no genuinely \emph{decentralized} approach, in which a shield onboard each agent can modify only the corresponding agent's behavior, exists ~\cite{BharNFM}. In such an approach, there would be no entity that has global information on the state space of the entire system. 

Without global information of the state, guaranteeing safety is, in general, undecidable~\cite{Schewe08}. Thus, we focus solely on enforcing \emph{local} safety properties, which is a subset of general temporal safety properties. A safety property is local if it can be enforced in the entire multi-agent system by enforcing it within each communication group. 
Essentially, shields are partial functions from the current states of the agents to the next states. 
Existing methods find this partial function by solving a reactive synthesis problem at design time~\cite{BloemKKW15,BharNFM,multiagentshield}. However, it is computationally prohibitive in the case of multi-agent systems since the resulting safety game has to take into account all possible behaviors from every agent \cite{BharNFM}. 

We formulate the synthesis of modified safe behavior of an agent as a graph search problem. More specifically, we assume agents know the intended behaviors of the other agents in its communication group and hence, an onboard enforcer can modify an agent's behavior, taking into account the behavior of the other agents in the same group. If the system continues to remain unsafe after the agent has changed its behavior, then the other agents are forced to change their behaviors. Thus, the synthesis of safe behavior for all agents in a communication group can be framed as a sequence of graph searches. This technique is similar to hierarchical path planning~\cite{Silver2005Jun}. In particular, we synthesize safe behavior online when required, i.e., when the intended trajectories of the agents violates a safety requirement. However, such an online approach to synthesizing new behaviors may create scenarios where some agents may never progress. That is, the behaviors of some agents may be perpetually modified to ensure safety.

%\Dj{It is a simple observation that from any combination of current states of the agents, it is possible to generate safe behavior for every agent (if one exists) using a sequence of graph searches. More specifically, since the agent in question knows the behaviors of the other agents in its communication group, the agent can modify its behavior, taking into account the behavior of the other agents in the same group.  Thus, the synthesis of safe behavior for an agent can be framed as a simple graph search. If the system continues to remain unsafe after the agent has changed its behavior, then the other agents are forced to change their behaviors.  This technique is similar to hierarchical path planning \cite{Silver2005Jun}. We consider the natural question: ``Is it possible to construct enforcers that synthesize safe-behavior by solving a sequence of simple graph searches in a decentralized fashion when required?'' However, such an online approach to synthesizing new behaviors may create scenarios where some agents may never progress. That is, the behaviors of some agents may be perpetually modified to ensure safety.}

In this paper, we present a novel \emph{decentralized} framework for \emph{online} synthesis for runtime enforcement. The enforcer onboard each agent issues modifications to the behavior of its corresponding agent in an order according to their \emph{priority} using graph search.
The framework uses a novel decentralized \emph{ordering mechanism} to dynamically maintain the agent's priorities to ensure that every agent can make progress according to their intended behaviors. We assume that the agents have agreed on this mechanism. Additionally, it is possible to compute the order between any two agents (the total order relation corresponding to the priorities) on the fly using only the flags that are local to the two agents. Moreover, only the corresponding agents can modify these flags. The presented ordering mechanism provably guarantees that every agent can acquire the highest priority in a finite length of time; hence live-locks are avoided. The online synthesis approach performs local behavior modification as needed, this circumvents the state-space explosion.

%\subsubsection*{Constrast with existing multi-agent path planning algorithms.} 
The proposed approach is similar to \emph{cooperative path planning} in multi-agent systems, which is a PSPACE-hard problem~\cite{Hopcroft1984Dec}. Hierarchical cooperative A* (HCA*) is a decentralized approach that uses fixed priorities on agents and makes a plan for an agent while respecting the plans of the agents with higher priorities~\cite{Silver2005Jun}. However, HCA* may require the agents to change their plan continuously and, therefore, cannot guarantee finite-time progress \cite{Silver2005Jun}. Proposed approaches that achieve completeness and produce optimal paths~\cite{Standley2010Jul,Standley2011} are either non-tractable or rely on global information. The method in \cite{Zhang2016} relaxes the reliance on global information; however, it still falls back to using it as a last resort. 

In contrast, the proposed framework does not need global information,  ensures bounded progress, and can be implemented in a decentralized manner. This level of decentralization, while ensuring system-level safety, is possible because the decentralized priority exchange mechanism we present ensures the absence of live-locks. In the existing techniques, live-locks have to be detected which requires global information. Additionally, if the agents in the system can idle, we provide a condition that guarantees \emph{completeness}. That is, if there exists a safe behavior then the enforcer can guarantee safety.

\subsubsection*{Contributions.} To our best knowledge, this paper presents the first approach where the enforcement of safety properties is viewed through the lens of cooperative path planning. The existing formulation for runtime enforcement through shielding uses reactive systems. However, this is unsuitable and cumbersome for the online approach.  Therefore, we provide a new formulation where the enforcers are tuples of partial functions. An extra benefit of such an approach is that the \emph{joint} behavior of all the agents can be directly expressed as a functional composition. Lastly, we prove the resulting enforcers also satisfy the following properties:
\begin{enumerate}
    \item \emph{Correctness}: The modified system behavior satisfies all the safety properties,
    \item \emph{Minimal Deviation}: The enforcer must modify behavior only if necessary and
    \item \emph{Bounded}: The deviation from the original behavior must be finite. We additionally show that the maximum deviation is linear in the number of agents.
    \item \emph{Completeness}: If a centralized stabilizing shield \cite{bloem2014sat} can guarantees correctness, the enforcers will also guarantee correctness. 
\end{enumerate}
By construction, the enforcers do not require global information. Additionally, we prove that the worst-case synthesis time for each agent is at most quadratic in the number of agents.

%\Dj{We compare decentralized runtime synthesis of shields to the centralized design-time synthesis of shields in~\cite{multiagentshield}. Shields are synthesized at runtime in a system with four agents and 25 states in under a second per agent compared to \cite{multiagentshield} which requires on the order of $10^4$ seconds for shield synthesis at design time.}
\section{Preliminaries}

$\Bool = \{\top,\bot\}$ is the domain of Booleans. \noindent A finite (infinite) word over a set $\Sigma$ of elements is a finite (infinite) sequence $w=a_1a_2 \dots a_n$ of elements of $\Sigma$. The length of $w$ is $|w|$. $\epsilon_\Sigma$ denotes the empty word over $\Sigma$ or $\epsilon$ when the context is clear. 
The concatenation of two words $w$ and $w'$ is denoted $w \cdot w'$. A word $w'$ is a prefix of a word $w$, denoted $w' \leq w$, whenever there exists a word $w''$ such that $w=w' \cdot w''$, and $w' < w$ if additionally $w'\neq w$. $w$ is said to be an extension of $w'$.  
The sets of all words and all non-empty words are denoted by $\Sigma^*$ and $\Sigma^+$, respectively. $\Sigma^{\leq k}$ denotes all words of length at most $k$. 
A language or a property over $\Sigma$ is any subset $L$ of $\Sigma^*$. 
%$L$ is \emph{prefix-closed} if all prefixes of all words from $L$ are also in $L: L=\{w~|~\exists w'\in L :w \leq w'\}$. 
%Similarly, a language $L$ over $\Sigma$ is \emph{extension-closed} if all extensions of all words from $L$ are in $L: L= \{w~|~\exists w' \in L: w'\leq w\}$. $w = w_0w_1\ldots$ be a possibly infinite word.  By $w[i:\ell]$, where $i,\ell \in \mathbb{N}$, we denote the sub-string $w_iw_{i+1}\ldots w_{i+\ell}$.

Let $G=(V, E)$ be a directed graph where $V$ is a finite set of nodes, and $E$ is a finite set of edges. The \emph{distance} $d(u,v)$ between two vertices $u$ and $v$ is defined as the length of a shortest directed path from $u$ to $v$. 
Let $\Agents$ denote a set of node labels and $\Actions$ denote a set of edge labels. $\mathbb{T}=\{1, 2, \dots \infty \}$ is a discrete set of time indices.
A graph with node labels and edge labels is called a \textit{labeled graph}. 
An \emph{edge labeling} is a function $\Elabel:E \times \Time \to \Actions$.
A  \emph{node labeling} is a function $\Vlabel: V \times \Time \to 2^\Agents$. A node labeling $\Vlabel$ is consistent at time $t$ if $\Vlabel$ partitions $V$, i.e., if for any $u$ and $v$, $\Vlabel(v,t) \cap \Vlabel(u,t) \neq \emptyset$ implies $u = v$. 

An \emph{environment} is a tuple $(G,\Elabel)$ where $G$ is a labeled graph and $\Elabel$ is an edge-labeling. $\Elabel(e,t)$ is the label of edge $e$ at time $t$. The environment is said to be \emph{static} if the associated edge labeling is time-invariant. 
%Unless explicitly specified we assume that the environment is static. 
For a static environment $(G,\Elabel)$, $\delta: V \times \Actions \to 2^V$ is called the \emph{transition function}. $\delta$ is \emph{deterministic}, if for any $v,v_1,v_2 \in V$ and $s \in \Actions$, $v_1,v_2 \in \delta(v,s)$ implies $v_1 = v_2$. 
The extended transition function $\hat{\delta}: V \times \Actions^* \to 2^V$ is defined recursively as $\hat{\delta}(v,\epsilon) = v$ and $\hat{\delta}(v,w\cdot a) = \delta(\hat{\delta}(v,w),a)$. A static environment is deterministic, if the associated transition function $\delta$ is deterministic.
A word $w = w_0w_1w_2\dots$ is said to \emph{induce a path} in a graph $G$ starting at vertex $v_0$ if there exists a sequence of vertices $v_1v_2v_3\dots$ such that $v_i \in \delta(v_{i-1},w_{i-1})$. In a static deterministic environment, the \emph{final state} induced by a finite word $w = w_0w_1\dots w_n$ starting at $v$ is the node $\hat{\delta}(v,w)$.

A \emph{trajectory} $p$ in a static environment $(G,\Elabel)$ is a pair $(v,w)$ where $w \in \Actions^*$ is a finite word such that $w$ induces a path in $G$ starting at vertex $v$. The final state of a trajectory $p = (v,w)$ is the final state induced by $w$ on $v$ and is given by $\hat{\delta}(v,w)$.  The concatenation of a trajectory $p$ and a word $w'$ is $p\cdot w' = (v,w\cdot w')$.  For trajectory $p$, we denote its \emph{sub-trajectory} $(v,w[i:\ell])$ by $p[i:\ell]$. 
A \emph{joint trajectory} is a finite set of trajectories. 
 %If the label of $v$ is $\{1\}$, then label of $v$ following $w$ is $\emptyset$ and the label of the vertex $\hat{\delta}(v,w)$ is $\{1\}$.

For any vertex label $u$, $p=(v,w)$ is a trajectory for $u$ at time $t$ if $u \in \Vlabel(v,t)$ and $w$ induces a path from $v$. Define the final state of $u$ through $p$ as the final state of $(v,w)$. The final state of $u$ through trajectory $p = (v,w)$ is ($\hat{\delta}(v,w)$) and is denoted $u \dhxrightarrow[v]{w} \hat{\delta}(v,w)$.
Let $\Vlabel(v,t) = \{u\} \cup X$, $(v,w)$ be a trajectory and $\Vlabel(\hat{\delta}(v,w),t) = Y$. If agent $u$ follows trajectory $(v,w)$, then at time $t+|w|$ the vertex label of $v$ is $X$ and the vertex label of $\hat{\delta}(v,w)$ is $Y \cup \{u\}$.  
 
Given  an $n$-tuple of  symbols $e =  (e_1,\dots,e_n)$,  for $i \in [1,n],\prod_i (e)$ is the projection of $e$ on its $i$-th element denoted $(\prod_i e \defeq e_i)$. $\Replace_i(e,x)$ replaces the i$^{th}$ element of $e$ with $x$ ,i.e., $\Replace_i(e,x) = (e_1,\dots,e_{i-1},x,e_{i+1},\dots,e_n)$.

\section{Online Enforcers}
\label{sec:formulation}

\subsubsection*{Environment and Agents.}
We model the region of operation of the agents as a deterministic environment $(G,\Elabel)$ with a consistent node labeling $\Vlabel$. The set of all agents is $\Agents$, and it is the same as the set of all node labels. At time $t$, an agent $u$ is said to be at location $v$ if $u \in \Vlabel(v,t)$. An agent can move from a vertex $v$ to a vertex $v'$ in one time unit through an action $s$, if there is an edge with label $s$ between $v$ and $v'$. 
 If $(v,w)$ is a trajectory, and agent $u$ follows the trajectory starting at time $t$, then at time $t+i$, $(i \leq w)$ the state of the agent $u$ is $\hat{\delta}(v,w[0:i])$. Furthermore, $\Vlabel(v,t) = \{u\}$ and $\Vlabel(\hat{\delta}(v,w[0:i]),t+i) = \emptyset$. 
For any trajectory $(v,w)$, we drop the initial vertex when it is clear.
At time $t+i$,  $\Vlabel(v,t) = \emptyset$ and $\Vlabel(\hat{\delta}(v,w[0:i]),t+i) = \{u\}$.  
Boolean $goal_{u,t}$ is true when agent $u$ reaches its final state at time $t$, which is referred to as the agent having \emph{completed the goal}.
Formally,
$$
goal_{u,t} = \begin{cases}
\top \text{ if } u \text{ has reached its final state} \\ \text{\quad following the trajectory } (v,w),
\\
\bot \text{ otherwise.}
\end{cases}
$$

\noindent  Associated with any agent is a unique priority from $[1,|\Agents|]$, defined as $priority: \Agents \times \Time \to [1,|\Agents|]$, such that $priority(u_1,t) = priority(u_2,t)$ implies $u_1 = u_2$. At any time $t$, the priorities of the agents induce a total order $\prec_t$ among them. 
For agents $u_1$ and $u_2$ in $\Agents$, $u_1 \prec_t u_2$ if and only if $priority(u_1,t) < priority(u_2,t)$.

\begin{example}
In Figure \ref{fig:eg7}b, blue and green agents operate in a grid world. The vertices of the underlying labeled graph $G$ are the cells in the grid. There is an edge from a vertex to another, if they are adjacent in the grid (no diagonal edges).
The set $\Actions \defeq \{l,r,t,d\}$ of edge labels is the set of actions available to each agent. The set $\Agents = \{blue,green\}$ of vertex labels corresponds to the set of agents operating in the system. $\Vlabel((2,4),0) = \{green\}$ and $\Vlabel((4,2),0) = \{blue\}$, that is, blue agent is at $(4,2)$ and green agent is at $(2,4)$. The trajectory of the blue agent is $((4,2),lll)$ and that of the green agent is $((2,4),ddd)$. At time $t = 1$, the labeling function is $\Vlabel((2,3),1) = \{green\}$, $\Vlabel((2,3),1) = \{blue\}$, $\Vlabel((2,4),1) = \emptyset$, $\Vlabel((4,2),1) = \emptyset$.  The blue and the green agents have reached their goals at $t=3$. Therefore, $goal_{blue,3} = \top$ and $goal_{blue,2} = goal_{blue,1} = \bot$. The final state of the blue agent is $(2,1)$ and the green agent is$(2,1)$. That is, 
\end{example}
\begin{small}
\begin{align*}
    blue \dhxrightarrow[(4,2)]{lll}(1,2)~\text{and}~
    green \dhxrightarrow[(2,4)]{ddd}(2,1).
\end{align*}
\end{small}

\subsubsection*{Communication.}
The agents in the system can communicate when they are close to each other. Moreover, two agents $u_i$ and $u_j$ can also communicate if there is a sequence of agents $c_1\dots c_k$ such that $c_1$ is $u_i$, and $c_k$ is $u_j$ and there is a path of length less than or equal to $d$ between agents $c_{i}$ and $c_{i+1}$. Here $d$ is a positive integer referred to as the \emph{communication constant}. 
%At any time, there can be at most $|\Agents|$ such groups. 
At any time $t$, $\Agents_i(t)$ denotes the communication group of agent $u_i$. The agents in the same communication group know the partial trajectories of the other agents in the group upto length $\ell$. Formally, every agent $u$ knows the partial trajectory $(v_{u'},w_{u'})[0:\ell]$ of every other agent $u'$ in its communication group, where $(v_{u'},w_{u'})$ is the trajectory for agent $u'$. Henceforth, $\ell$ is referred to as the \emph{look-ahead}.
 
\begin{figure*}
\centering
\definecolor{green1}{rgb}{0,0.6,0}
\subfloat[Communication groups]{\begin{tikzpicture}[line join=round,x=1cm,y=1cm]
\begin{axis}[
scale=0.6,
x=0.5cm,y=0.5cm,
axis lines=middle,
ymajorgrids=true,
xmajorgrids=true,
xmin=0,
xmax=10,
ymin=0,
ymax=10,
xtick={-2,-1,...,25},
ytick={-2,-1,...,14},]

\begin{scriptsize}
\draw [->,line width=0.6pt] (4,7) -- (7,7);
\draw [->,line width=0.6pt] (2,7) -- (2,9.5);
\draw [->,line width=0.6pt] (2,5) -- (2,2);
\draw [line width=0.6pt,dashed] (3,6) circle (0.8cm);
\draw [line width=0.6pt,dashed] (6,2) circle (0.6cm);
\draw [->,line width=0.6pt] (6,1) -- (9,1);
\draw [->,line width=0.6pt] (6,3) -- (3,3);

\draw [fill=blue] (2,7) \Square{1.3pt};
\draw [fill=green1] (4,7) \Square{1.3pt};
\draw [fill=red] (2,5) \Square{1.3pt};
\draw [fill=black] (6,3) \Square{1.3pt};
\draw [fill=purple] (6,1) \Square{1.3pt};

\draw[color=blue] (2,6.3) node {t=0};
\draw[color=green1] (4,6.3) node {t=0};
\draw[color=red] (3.3,5) node {t=0};
\draw[color=purple] (5,1.5) node {t=0};
\draw[color=black] (6,2.5) node {t=0};

\end{scriptsize}
\end{axis}
\end{tikzpicture}}
\subfloat[Unsafe behavior]{\begin{tikzpicture}[line join=round,x=1cm,y=1cm]
\begin{axis}[
scale=0.6,
x=1cm,y=1cm,
axis lines=middle,
ymajorgrids=true,
xmajorgrids=true,
xmin=0,
xmax=5,
ymin=0,
ymax=5,
xtick={-3,-2,...,35},
ytick={-7,-6,...,17},]

\draw [->,line width=0.6pt] (2,4) -- (2,1);
\draw [dashed,->,line width=0.6pt] (4,2) -- (1,2);

\begin{scriptsize}

\draw [fill=green1] (2,4) circle (1.3pt);
\draw[color=green1] (1.5,4.3) node {t=0};

%\draw [fill=green1] (2,3) circle (1.3pt);
%\draw[color=green1] (1.5,3.3) node {t=1};

\draw [fill=green1] (1.8,2.2) circle (1.3pt);
\draw[color=green1] (1.5,2.5) node {t=2};

%\draw [fill=green1] (2,1) circle (1.3pt);
%\draw[color=green1] (1.5,0.8) node {t=3};

\draw [fill=blue] (4,2) \Square{1.3pt};
\draw[color=blue] (4,2.5) node {t=0};

%\draw [fill=blue] (3,2) \Square{1.3pt};
%\draw[color=blue] (3,2.5) node {t=1};

\draw [fill=blue] (2,2) \Square{1.3pt};
\draw[color=blue] (2.5,1.6) node {t=2};

%\draw [fill=blue] (1,2) \Square{1.3pt};
%\draw[color=blue] (1,1.6) node {t=3};
\end{scriptsize}
\end{axis}
\end{tikzpicture}}
\subfloat[Safe behavior]{\begin{tikzpicture}[line join=round,x=1cm,y=1cm]
\begin{axis}[
scale=0.6,
x=1cm,y=1cm,
axis lines=middle,
ymajorgrids=true,
xmajorgrids=true,
xmin=0,
xmax=5,
ymin=0,
ymax=5,
xtick={-3,-2,...,35},
ytick={-7,-6,...,17},]

\draw [dashed,line width=0.6pt] (4,2)-- (3,2);
\draw [dashed,line width=0.6pt] (3,2)-- (3,3);
\draw [dashed,line width=0.6pt] (3,3)-- (2,3);
\draw [dashed,line width=0.6pt] (2,3)-- (1,3);

\draw [dashed,line width=0.6pt,->] (1,3) -- (1,2);
\draw [->,line width=0.6pt] (2,4) -- (2,1);

\begin{scriptsize}
\draw [fill=green1] (2,4) circle (1.3pt);
\draw[color=green1] (1.5,4.3) node {t=0};

%\draw [fill=green1] (1.8,3.2) circle (1.3pt);
%\draw[color=green1] (1.5,3.5) node {t=1};

%\draw [fill=green1] (2,2) circle (1.3pt);
%\draw[color=green1] (1.5,2.3) node {t=2};

\draw [fill=green1] (2,1) circle (1.3pt);
\draw[color=green1] (1.5,0.5) node {t=3};

\draw [fill=blue] (4,2) \Square{1.3pt};
\draw[color=blue] (4,2.4) node {t=0};

%\draw [fill=blue] (3,2) \Square{1.3pt};
%\draw[color=blue] (3,1.6) node {t=1};

%\draw [fill=blue] (3,3) \Square{1.3pt};
%\draw[color=blue] (3.5,3.3) node {t=2};

%\draw [fill=blue] (2,3) \Square{1.3pt};
%\draw[color=blue] (2.5,3.3) node {t=3};

%\draw [fill=blue] (1,3) \Square{1.3pt};
%\draw[color=blue] (0.6,3.3) node {t=4};

\draw [fill=blue] (1,2) \Square{1.3pt};
\draw[color=blue] (1,1.7) node {t=5};
\end{scriptsize}
\end{axis}
\end{tikzpicture}}
\caption{(a) The communication groups at time $t=0$. The communication constant $d$ is $2$. The agents in the same group have been encircled. The black and purple agents are in a group. While blue, green and the red agents are in a different group. (b) Grid world example: There are two agents (blue and green). Their intended trajectories are marked by lines. Their positions at different times are also shown. At time $t=2$ the blue and green agents occupy the same cell, hence $\Safety(2) = \bot$. However, the system is still safe at times $t=0$ and $t=1$. (c)  The modified trajectory for the blue agent as a consequence of the enforcer $S(blue,0)$ on the blue agent at $t=0$ is shown on the right.}
\label{fig:eg7}
\end{figure*}
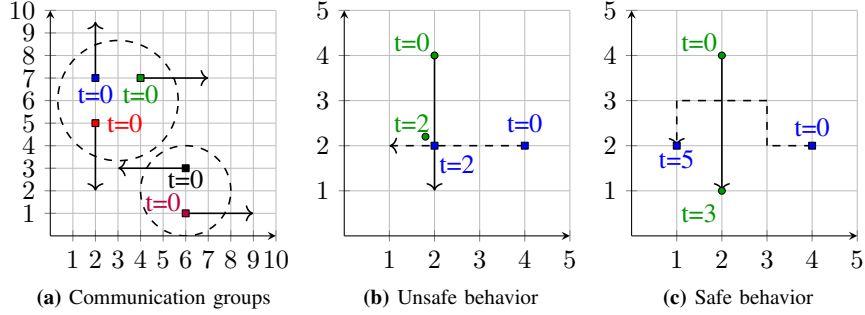
\begin{example}
In Figure \ref{fig:eg7}a, the trajectories of 5 agents are shown. The communication constant $d = 2$ and $\ell = 3$. 
At time $0$, the $\Agents_{blue}(0)=\Agents_{green}(0)=\Agents_{red}(0)= \{blue,green,red\}$ and $\Agents_{purple}(0) = \Agents_{black}(0) = \{purple,black\}$. All agents in the first group know that the trajectory of the red agent is $((2,5),ddd)$, the blue agent is $((2,7),tt)$ and the green agent is $((4,7),lll)$. In the second group, all agents know that the trajectory of the purple agent is $((6,3),ll)$ and the black agent is $((6,1),rr)$. 
\end{example}

\subsubsection*{Safety Functions.}
A \emph{safety} property $\Safety$ is a function $\Safety: 2^P \to \Bool$ from the vertex labeling to Booleans, where $P = V \times 2^\Agents$.
As the vertex labels themselves depend on time, we extend the notion of safety to safety at time $t$. If the vertex labeling at time $t$ is safe, then $\Safety(t) = \top$.
\begin{example}
Consider a safety function $\Safety:  \Time \to \Bool$ defined as
\[   
\Safety(t) = 
     \begin{cases}
       \top & \text{if for any } v \in V~ |\Vlabel(v,t)|=1 \\& \text{and } \Vlabel(u,t) = \Vlabel(v,t) \text{ implies } u = v, \\
       \bot & \text{otherwise.} \\
    \end{cases}
\]
Simply if $X$ is the vertex labeling at some time, $X$ is \emph{safe} if there is only one agent at any location at any given time and any agent can be present at only one location at any time (consistent label). 
In the grid world in Figure \ref{fig:eg7}c, the system is safe at all times. In the grid world in Figure \ref{fig:eg7}b, the system is unsafe at time $t = 2$ as the blue and the green agents occupy the same location.
\label{eg:safety_fn}
\end{example}
%\vspace{-5mm}
Next we extend the notion of safety to trajectories. For every agent $u$, denote its trajectory by $(v_u,w_u)$ . The system is \emph{safe on the trajectories}, if the system is safe at all times. Formally, $\Safety(\Agents) = \top$ if for all $t \in [0,m]$ ,$\Safety(t) = \top$, where $m = \max \{|w_u|:u \in \Agents \}$.

\begin{example}
Consider safety as described in Example 3. Then, in Figure \ref{fig:eg7}c, the system is safe on trajectories $ltlld$ and $ddd$ for the blue and green agents respectively. However, in Figure \ref{fig:eg7}b, the system is not safe as it is unsafe at $t=2$. We say that the blue and green agents violate safety.
\end{example}

Recall the agents in the system have only limited communication. Therefore, we are interested in a subclass of safety functions that can be enforced across the system by enforcing them locally in every communication group. 
Suppose $\Safety_i$ is a safety property such that $\Safety_i(t) = \top$ if and only if the agents in the communication group $\Agents_i$ are safe.
%For every communication group} $\Agents_i$, $\Safety_i$ be a safety property that the agents in group $\Agents_i$ are safe. 
Then, the property $\Safety(t)$ defined as $$ \Safety(t) \defeq \bigwedge_{i \in 1..|\Agents|} \Safety_i(t)$$ is a \emph{local} safety property.
%A safety property $\Safety$ is a \emph{local safety property} if  $$\Safety(t) = \top \Leftrightarrow  \forall i \in [1,|\Agents|] ~\Safety_i(t) = \top.$$ 
Observe that the safety function defined in Example 3 is a local safety property.

\subsubsection*{Enforcers.}
Informally, the purpose of an enforcer is to take a (possibly incorrect) trajectory produced by a running system and to transform it into a trajectory that is safe with respect to a local safety function $\Safety$ that we want to enforce. Abstractly, an enforcer can be seen as a function that transforms trajectories. 

Denote by $S(u,t)$ the enforcer acting on agent $u$ at time $t$. $S(u,t)$ is a pair of \emph{partial} functions $\langle S_1(u,t),S_2(u,t) \rangle$. $S_1(u,t)$ accepts a finite trajectory for each agent in the system and returns a modified trajectory for agent $u$. $S_2(u,t)$ accepts a vector of current priorities and a vector of Booleans $goal_{u,t}$ for the agents in the system. It returns a vector of priorities with only the priority of its corresponding agent $u$ possibly changed. Formally, the enforcer on agent $u$ at time $t$ is a pair of partial functions $\langle S_1(u,t),S_2(u,t) \rangle$ such that
\begin{align*}
    &S_1(u,t) : \Lang_1 ^{|\Agents|} \to \Lang_1^{|\Agents|} \text { and } \\
    &S_2(u,t) : [1,|\Agents|]^{|\Agents|} \times \mathbb{B}^{|\Agents|} \to [1,|\Agents|]^{|\Agents|} \times \mathbb{B}^{|\Agents|}, 
\end{align*}
\noindent where $\Lang_1 = [1,|\Agents|] \times V \times \Actions^*$. The above definition of a enforcer is quite general as both the input trajectory and the modified trajectory can be of arbitrary length. Next, we introduce $(\ell,\ell')$-enforcers. For every agent in the system, these enforcers accept trajectories of length at most $\ell$ and return a new trajectory of length at most $\ell' \geq \ell$.  
Formally, an $(\ell,\ell')$-enforcer 
%\footnote{In the rest of this paper, we refer to $(\ell,\ell')$-enforcers as just enforcers.} 
on agent $u$ at time $t$ is a pair of partial functions $\langle S_1(u,t),S_2(u,t) \rangle$ such that 
\begin{align*}
    &S_1(u,t) : \Lang_I ^{|\Agents|} \to \Lang_O ^{|\Agents|} \text{ and } \\ 
    &S_2(u,t) : [1,|\Agents|]^{|\Agents|} \times \mathbb{B}^{|\Agents|} \to [1,|\Agents|]^{|\Agents|} \times \mathbb{B}^{|\Agents|}, 
\end{align*}
where $\Lang_I = [1,|\Agents|] \times V \times \Actions^{\leq \ell}$ and $\Lang_O = [1,|\Agents|] \times V \times \Actions^{\leq \ell'}.$

\begin{example}
In Figure \ref{fig:eg7}b, the blue and green agents occupy the same location at time $t=2$. However, in Figure \ref{fig:eg7}c, the blue agent's trajectory has been modified by the enforcer onboard. As a result they never occupy the same position at the same time. Priority of the blue agent is 1 and the priority of the green agent is 2. $S_1(blue,0)\big((blue,1,lll)(green,2,ddd)\big) = (blue,1,ltlld)(green,2,ddd)$.
\end{example}

\subsubsection*{Composition of enforcers.}

When multiple agents act in the same system, their trajectories and their priorities are modified only by their respective enforcers. However, the individual enforcers act together to make the system safe. The joint behavior of enforcers are captured by functional composition.  We first define \emph{composition of enforcers} for two agents $u_1$ and $u_2$ at time $t$. If $u_1 \prec_t u_2$ then
\begin{align*}
    S(u_1,t) \circ S(u_2,t) &= S(u_2,t) \circ S(u_1,t) \\
    &=\langle S_1(u_2,t) \circ S_1(u_1,t), S_2(u_2,t) \circ S_2(u_1,t)\rangle.
\end{align*}
%$$S(u_1,t) \circ S(u_2,t) = S(u_2,t) \circ S(u_1,t) =  \langle S_1(u_2,t) \circ S_1(u_1,t), S_2(u_2,t) \circ S_2(u_1,t)\rangle.$$
This composition can be extended to an arbitrary number of enforcers by composing their constituent functions in the order $\prec_t$.

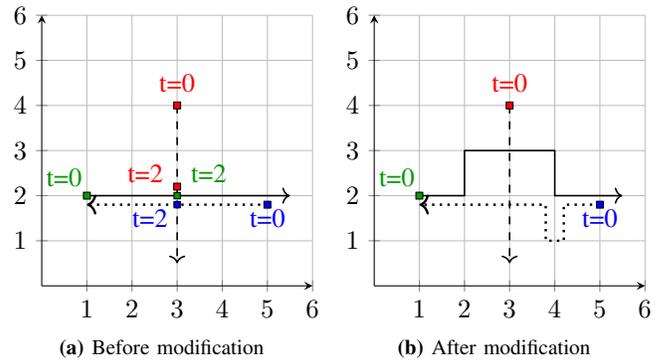
\begin{figure}[htb!]
\definecolor{green1}{rgb}{0,0.6,0}
\centering
\subfloat[Before modification]{
\begin{tikzpicture}[line join=round,x=1cm,y=1cm]
\begin{axis}[
scale=0.6,
x=1cm,y=1cm,
axis lines=middle,
ymajorgrids=true,
xmajorgrids=true,
xmin=0,
xmax=6,
ymin=0,
ymax=6,
xtick={0,1,...,5,6},
ytick={0,1,...,6},
]
\draw [dashed,->,line width=0.6pt] (3,4) -- (3,0.5);
\draw [->,line width=0.6pt] (1,2) -- (5.5,2);
\path[draw,dotted,<-,line width=0.9pt] (1,1.8) -- (2,1.8) -- (3,1.8) -- (4,1.8) -- (5,1.8);
\begin{scriptsize}
\draw [fill=red] (3,4) \Square{1.3pt};
\draw [fill=green1] (1,2) \Square{1.3pt};
\draw [fill=blue] (5,1.8) \Square{1.3pt};

\draw [fill=green1] (3,2) \Square{1.2pt};
\draw [fill=blue] (3,1.8) \Square{1.2pt};
\draw [fill=red] (3,2.2) \Square{1.3pt};

\draw[color=red] (3,4.5) node {t=0};

\draw[color=red] (2.3,2.5) node {t=2};

\draw[color=blue] (2.4,1.5) node {t=2};
\draw[color=blue] (5,1.5) node {t=0};

\draw[color=green1] (0.5,2.4) node {t=0};
\draw[color=green1] (3.7,2.5) node {t=2};
\end{scriptsize}
\end{axis}
\end{tikzpicture}}
\subfloat[After modification]{
\begin{tikzpicture}[line join=round,x=1cm,y=1cm]
\begin{axis}[
scale=0.6,
x=1cm,y=1cm,
axis lines=middle,
ymajorgrids=true,
xmajorgrids=true,
xmin=0,
xmax=6,
ymin=0,
ymax=6,
xtick={-1,0,...,30},
ytick={-1,0,...,18},
]

\draw [->,dashed,line width=0.6pt] (3,4) -- (3,0.5);
\path[draw,dotted,->,line width=0.9pt] (5,1.8) -- (4.2,1.8)-- (4.2,1)-- (3.8,1) -- (3.8,1.8)--(1,1.8);
\path[draw,->,line width=0.6pt] (1,2)--(2,2)--(2,3)--(3,3)--(4,3)--(4,2)--(5,2)--(5.5,2);

\begin{scriptsize}
\draw [fill=red] (3,4) \Square{1.3pt};
\draw [fill=green1] (1,2) \Square{1.3pt};
\draw [fill=blue] (5,1.8) \Square{1.3pt};

\draw[color=red] (3,4.5) node {t=0};
\draw[color=blue] (5,1.5) node {t=0};
\draw[color=green1] (0.5,2.4) node {t=0};

\end{scriptsize}
\end{axis}
\end{tikzpicture}}
\caption{Original trajectories of the agents and their modified trajectories. The system is safe after the modification of the trajectories by the enforcers. (a) Blue, green and red agents violate safety at time 2 in (3,2). (b) The trajectories of blue and green agents have been modified by the respective enforcers. However, red agent does not deviate.}
\label{fig:eg4}
\end{figure}

\begin{example}
In Figure \ref{fig:eg4}a, at $t=0$, blue agent has priority 1, green agent has priority 2 and red agent priority 3. All the three agents occupy the same location (3,2) at time 2. The enforcer on the blue agent modifies its trajectory first and this modification is relayed to the green and the red agents. The enforcer on the green agent then modifies its trajectory. There is no change in the trajectory of the red agent as now there is no safety violation. The modified trajectories are shown in Figure \ref{fig:eg4}b.
\end{example}

\subsubsection*{Properties of enforcers.}
We define the desired properties for a set of enforcers. 
%Let the set $\{p_{u}~|~u \in \Agents\}$ be the trajectories of the agents at time $t$. 
For any agent $u$ in $\Agents$, let $S(u,t)$ be the enforcer acting on agent $u$ at time $t$, $p_u$ its original trajectory at time $t$ and $p'_u$ its modified trajectory at time $t$.  
%\Dj{Let} $S(u_1,t),\dots,S(u_n,t)$ be the enforcers on the agents $u_1\ldots u_n$, respectively \Dj{and} $ \{p'_{u_1},p'_{u_2},\dots,p'_{u_n}\}$ be \Dj{their} modified trajectories as a consequence of the enforcers. 
The enforcers $\{S(u,t)~|~u \in \Agents \}$ are \emph{correct} if the modified trajectories 
$\{p'_{u}~|~ u \in \Agents\}$ 
are safe and the final states are unchanged. 
An enforcer $S(u,t)$ is said to \emph{cause minimum deviation} if $p'_u = p_u$ when $\Safety_u(t) = \top$. 
The enforcers $\{S(u,t)~|~u \in \Agents\}$ are \emph{bounded}, if there exists $\ell$ and  $\ell'$ in $\mathbb{N}$ such that all the enforcers are $(\ell,\ell'$)-enforcers. 
We later prove that boundedness and correctness ensure that all agents progress in finite time, while still guaranteeing the safety of the system. 
We now state the problem studied in this paper.

\begin{prob}
Given a set $\Agents = \{u_1,\dots,u_n\}$ of agents and a set $\{(v_u,w_u)~|~u \in \Agents \text{ and } |w_u| \leq \ell \}$ of their trajectories, 
construct a set $\{S(u,t)~|~u \in \Agents \text{ and } t \in \Time\}$ of enforcers such that these enforcers are correct, cause minimum deviation and bounded. 
\end{prob}

\section{Online Synthesis}
Informally, the enforcer $S(u,t)$ onboard agent $u$ can directly affect only the trajectory of agent $u$. 
Every enforcer has access to a \emph{pathfinder} that modifies the corresponding trajectory. If a potential safety violation is detected, the enforcer on the agent with the lowest priority calls the pathfinder first. The order $\prec_t$ determines the next agent potentially required to modify its trajectory.
The pathfinder resolves conflicts, if any, within the group. If a new agent comes into the group, the pathfinder is called by the enforcer on the lowest priority agent. The trajectory of an agent is not modified if it is not involved in a safety violation. 
A \emph{ordering mechanism} maintains the order $\prec_t$ among the agents. When an agent reaches its final state, then its intended trajectory is updated and the ordering mechanism also updates the priorities.

\subsubsection*{Pathfinder.}

Informally, the pathfinder returns a new path whenever called. It constructs a graph and searches for a path in the graph from a vertex corresponding to the current location to a vertex in a target set corresponding to the agent's final state. This graph does not have any outgoing edges from vertices that correspond to unsafe configurations. After a single call to the pathfinder, the maximum length of the modified trajectory is at most $\ell + k$. where $\ell \leq k < d$ is some constant. But such a path may not always exist. In this case, the pathfinder returns a path that at the minimum keeps the system safe.

\begin{assume}
The graph $G$ is 2-edge connected and there is self loop on every vertex in $G$.
\end{assume}

%The alternate target set ensures safety and the that the agents is at most $k$ steps away from its final state. 

For any agent $u \in \Agents$, its trajectory $p_u$ is $p_u=(v_u,w_u)$.
%such that \Dj{$|w_u| \leq \ell \leq d$}, i.e., the next goal for any agent is some state that is visible to it. 
The final state of agent $u$ is $v_u^f$ and $priority(u,t)$ is its priority at time $t$. 

%Formally, 

\begin{example}
Figure \ref{fig:eg5} depicts the graph $G_{blue}^0$ constructed by the pathfinder on the blue agent for the example in Figure \ref{fig:eg4}a. The initial position is 
$v_{init} = ((4,2),0)$ and the target set $F = \{((1,2),4),((1,2),5)),((1,2),6))\}$ is marked red. The nodes occupied by some other higher priority agent at the time are marked by black circles. These black nodes do not have any out-edges. The pathfinder returns a path from $v_{init}$ to some vertex in $F$. The positions of the green agent are unmarked as it has a lower priority.
\begin{figure}[!htb]
    \centering
    \begin{tikzpicture}[scale=0.8]
    \draw [dashed](0,0) ellipse (0.5cm and 3cm);
    \draw [dashed](1.5,0) ellipse (0.5cm and 3cm);
    \draw [dashed](3,0) ellipse (0.5cm and 3cm);
    \draw [dashed](4.5,0) ellipse (0.5cm and 3cm);
    \draw [dashed](6,0) ellipse (0.5cm and 3cm);
    \draw [dashed](7.5,0) ellipse (0.5cm and 3cm);
    \draw [dashed](9,0) ellipse (0.5cm and 3cm);
    \draw (0,-3.5) node {$t=0$};
    \draw (1.5,-3.5) node {$t=1$};
    \draw (3,-3.5) node {$t=2$};
    \draw (4.5,-3.5) node {$t=3$};
    \draw (6,-3.5) node {$t=4$};
    \draw (7.5,-3.5) node {$t=5$};
    \draw (9,-3.5) node {$t=6$};
    
    \draw (0,2) circle (2pt) node [anchor=north] {0,0};
    \draw (0,0.8) node {$\vdots$};
    \draw [fill = blue] (0,0) circle (2pt) node [anchor=north] {4,2};
    %\draw (0,-1.2) circle (1pt);
    \draw (0,-1) node {$\vdots$};
    %\draw (0,-0.8) circle (1pt);
    \draw (0,-2) circle (2pt) node [anchor=north] {5,5};
   
    \draw (1.5,2) circle (2pt) node [anchor=north] {0,0};
    \draw [fill = black] (1.5,0.5) circle (2pt) node [anchor=north] {2,3};
    %\draw [fill = black] (1.5,1) circle (2pt) node [anchor=north] {1,2};
    \draw (1.5,-2) circle (2pt) node [anchor=north] {5,5};
    
    \draw (1.5,-0.9) node {$\vdots$};
    %\draw (1.5,-1.2) circle (1pt);
    %\draw (1.5,-1) circle (1pt);
    %\draw (1.5,-0.8) circle (1pt);
    
    \draw (3,2) circle (2pt) node [anchor=north] {0,0};
    \draw [fill=black] (3,0.5) circle (2pt) node [anchor=north] {2,2};
    \draw (3,-2) circle (2pt) node [anchor=north] {5,5};
    \draw (3,-1) circle (2pt) node [anchor=north] {i,j};

    \draw (4.5,0.3) circle (2pt) node [anchor=north] {i-1,j};
    \draw (4.5,-0.2) circle (2pt) node [anchor=north] {i,j-1};
    \draw (4.5,-0.8) circle (2pt) node [anchor=north] {i,j+1};
    \draw (4.5,-1.5) circle (2pt) node [anchor=north] {i+1,j};
    
    \draw  (3.08,-1) -- (4.45,0.28);
    \draw  (3.08,-1) -- (4.45,-0.2);
    \draw  (3.08,-1) -- (4.45,-0.8);
    \draw  (3.08,-1) -- (4.45,-1.5);

    \draw (4.5,2.5) circle (2pt) node [anchor=north] {0,0};
    \draw [fill=black](4.5,0.8) circle(2pt) node[anchor = north] {3,2};
    \draw [fill=black](4.5,1.5) circle(2pt) node[anchor = north] {2,1};
    \draw (4.5,-2.8) circle (2pt) node [anchor=south] {5,5};
    
    \draw [fill=red] (6,0) circle (2pt)node [anchor=north] {1,2};
    \draw (6,2) circle (2pt) node [anchor=north] {0,0};
    \draw [fill=black](6,0.6) circle(2pt) node[anchor = north] {3,1};
    %\draw [fill=black](6,1.2) circle(2pt) node[anchor = north] {2,1};
    \draw (6,-2) circle (2pt) node [anchor=north] {5,5};
    
    %\draw (6,-1.2) circle (1pt);
    \draw (6,-1) node {$\vdots$};
    %\draw (6,-0.8) circle (1pt);

    \draw [fill=red] (7.5,0) circle (2pt)node [anchor=north] {1,2};
    \draw (7.5,2) circle (2pt) node [anchor=north] {0,0};
    \draw [fill=black](7.5,0.6) circle(2pt) node[anchor = north] {3,1};
    %\draw [fill=black](7.5,1.2) circle(2pt) node[anchor = north] {2,1};
    \draw (7.5,-2) circle (2pt) node [anchor=north] {5,5};
    
    %\draw (7.5,-1.2) circle (1pt);
    \draw (7.5,-1) node {$\vdots$};
    %\draw (7.5,-0.8) circle (1pt);

    \draw [fill=red] (9,0) circle (2pt)node [anchor=north] {1,2};
    \draw (9,2) circle (2pt) node [anchor=north] {0,0};
    \draw [fill=black](9,0.6) circle(2pt) node[anchor = north] {3,1};
    %\draw [fill=black](9,1.2) circle(2pt) node[anchor = north] {2,1};
    \draw (9,-2) circle (2pt) node [anchor=north] {5,5};
    %\draw (9,-1.2) circle (1pt);
    \draw (9,-1) node {$\vdots$};
    %\draw (9,-0.8) circle (1pt);
    
    \end{tikzpicture}
    \caption{The graph constructed by the pathfinder for the blue agent in Example 5. The set of vertices is $\{0,\dots5 \} \times \{0,\dots5 \} \times  \{ 0,\dots,6\}$. There is an edge between $((i,j),t_1)$ and $((i',j'),t_2)$ if and only if $|i'-i|+|j'-j| = 1$ and $t_2 = t_1 + 1$ and there is no out-edge from a black vertex (a black vertex corresponds to it being occupied by some agent at the corresponding time). $v_{init}$ is blue and the vertices in $F$ are red. The positions of the green agent are unmarked as it has a lower priority than the blue agent.}
    \label{fig:eg5}
\end{figure}
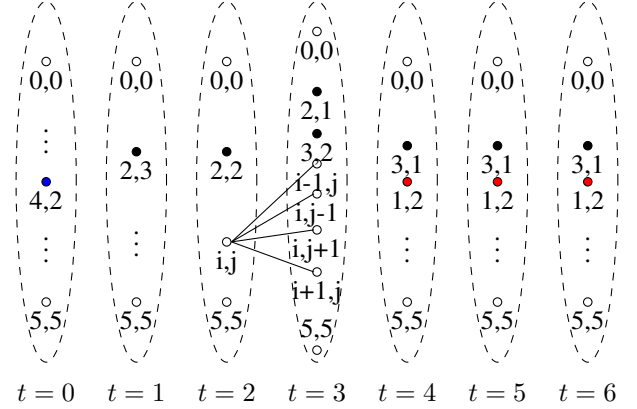
\end{example}

%In \Dj{Example}~\ref{eg:collavoid}, we describe the graph and the initial and target states for the safety function corresponding to collision-avoidance \Dj{from} Example 3 for some agent \Dj{$u$.}
\begin{example}\label{eg:collavoid}
For the safety function defined in Example 3, the 
pathfinder constructs the graph $G^t_{u} = (V', E')$ where $V'=V\times[t,t+d(v_u,v_u^f)+k]$. There is an edge with label $e$ between $(v,t)$ and $(v',t' +1)$, if in $G$ there is an edge $(v,v')$ with label $e$ and $\Vlabel(v,t') = \emptyset$, i.e., there are no higher priority agents occupying the same state. 
The target set $F$ is $\{(v^f_{u} ,t') | t+ d(v_u,v_u^f) \leq t' \leq t + d(v_u,v_u^f) + k\}$ and the initial state is $v_{init} = (v_{u},t)$. %\Dj{ The alternate target set 
%$F' = \{(v,t') | \Vlabel(v,t') = \emptyset~,t+ \ell \leq t' \leq t + \ell + k \}.$ 
%$F' = \{(v,k) | \Vlabel(v,k) = \emptyset~\text{and there is a path of length at most}~k~\text{between}~v~\text{and} ~v_u^f \}.$
%$F' = \{(v,t+1) | \Vlabel(v,t+1) = \emptyset \}.$
%The alternate target set corresponds to all the states that are not occupied by agents with higher priorities.
\end{example}

\subsubsection*{Occupancy Graph.}
The occupancy graph $O^t_u$ is similar to the pathfinder graph, however none of the edges are removed from the occupied states. The occupied states are labeled with the corresponding agents. Formally, the occupancy graph $O^t_{u} = (V', E')$ where $V'=V\times[t,t+d(v_u,v_u^f)+k] \setminus \{(v,t')|~v~\text{ is occupied by highest priority agent at } t'\}$. There is an edge with label $e$ between $(v,t)$ and $(v',t' +1)$, if in $G$ there is an edge $(v,v')$ with label $e$. If $u' \in \Vlabel(v,t)$ and $priority(u) \geq priority(u')$, then $u' \in \Vlabel((v,t))$.  

\subsubsection*{General Pathfinder Graph.}

Next, we present the pathfinder construction for any local safety property $\Safety$ for agent $u$. For this, we need the safety function $\bar{\varphi}_u(t)$ defined as:
$$\bar{\varphi}_u(t) \defeq \bigwedge_{priority(u)\leq priority(i)} \Safety_i(t) .$$
The safety function $\bar{\varphi}_u(t)$ ensures that all the agents with priorities higher than $u$ are safe.
$G_{u}^t = (V',E')$ is the graph whose nodes $V'$ are $V' = P \times [t,t+d(v_u,v_u^f)+k]$. Recall $P = V \times 2^\Agents$. There is an edge between $(v,t_1)$ to $(v',t_2)$ if i) $\Safety_u(v)  = \top$, ii) $(v,v')$ is an edge in $G$, iii) $t_2 = t_1 + 1$, iv) $\bar{\varphi}_u(t) = \top$, and v) all the other higher priority agents are following their trajectories, i.e., \begin{align*}
        & u_i \in \Vlabel\left( \hat{\delta}(u_i^t,\prod_i W[0:t_1]),t_1 \right) \\
        \text{ and } 
        & u_i \in \Vlabel\left(( \hat{\delta}(u_i^t,\prod_i W[0:t_2]),t_2\right).
\end{align*}
   
The initial vertex $v_{init}$ is $(a,t)$ where $a = v_u \times \Vlabel(v_u,t)$ and the set of target vertices $F$ is
$F=\{ (v,j)|v \in P' \text{ and } t+d(v_u,v_u^f) \leq j \leq t+d(v_u,v_u^f)+k \}$,
where $P' \subseteq P$ is a subset of all vertex labeling such that agent $u$ has reached its final state. %The alternate target set 
%$F' = \{ (v,t')~|~\bar{\Safety}_u(t') = \top \land u \in \Vlabel(v,t') \land t+ \ell \leq t' \leq t + \ell + k \}$.
%$F' = \{ (v,\ell)|\bar{\Safety}_u(\ell)=\top\land u\in\Vlabel(v,\ell)~\text{and there is a path of length at most}~\ell~\text{between}~v~\text{and} \\~v_u^f \}$.
%$F' = \{ (v,t+1)|\bar{\Safety}_u(t+1)=\top\land u\in\Vlabel(v,t+1)~\text{and there is a path of length at most}~\ell~\text{between}~v~\text{and} \\~v_u^f \}$.
Next, we describe the working of the pathfinder. The pathfinder constructs the graph $G_u^t$, initializes  $v_{init}$, a set $F$ of target vertices. The pathfinder then returns a path from $v_{init}$ to some state in $F$ if it exists, else it returns a random path of length 1 that only ensures safety. If no such path exists, the agent finds a shortest path in the occupancy graph to a vertex that ensures safety. This path may have other agents. All the other agents are forced to move 1 step along this path. If $p = v_1 v_2 \dots v_i v_{i+1}$ is the path from the occupancy graph, and agent $m$ is at $v_i$, then the trajectory of $m$ is replaced by $(v_i,v_{i+1})$. In short, if some lower priority agent cannot plan around the higher agent, then it might disturb all the agents other than the highest priority agent to ensure safety. However, the highest priority agent's path cannot be changed.
In the worst case, the pathfinder ensures that the agent with the highest priority can progress without any modifications. In essence lower priority agents progress, if they can plan around the highest priority agent.
\begin{smaller}
\begin{algorithm}
\KwResult{Safe path for $u$}
Initialize $G_u^t$ and $O_u^t$\;
Initialize $v_{init}$ and $F$\;
$P$ = path in $G_u^t$ from $v_{init}$ to $F$\;
\If{$P$ exists}{return $p$\;}
\Else{
$P$ = shortest path to an unoccupied vertex in $O_u^t$\;
\ForAll{vertex $v_i \in P$}{
{
\If{$\Vlabel(v_i,t) = a$}{$w_a = (v_i,v_{i+1})$\;}
}
}

}
\caption{Pathfinder on agent $u$}
\label{algo:pathfinder}
\end{algorithm}
\end{smaller}

\subsubsection*{Ordering Mechanism.}
The priorities of the agents cannot remain static with time. Otherwise, some agent might be forced to change its trajectory infinitely often. 
%The $S_2$ component maintains the order among agents dynamically.
In the sequel, we present the ordering mechanism. 

\paragraph*{Overview of Ordering Mechanism for Two Agents.}

\noindent Consider a system with two agents $a$ and $b$ that have communicated, i.e., observed each other's trajectories. Agent $a$ maintains a flag $c_a^b$ and agent $b$ maintains a flag $c^a_b$. If agent $a$ has reached its final state after communicating with agent $b$, then the flag $c^b_a$ is set to 1. 
Suppose agent $b$ is yet to reach its final state and there is a safety violation after $a$ has completed its goal, then in order to ensure freedom from locks, agent $a$ is forced to modify its trajectory. When agent $b$ reaches its final state, both agents have uniformly completed their goals and the flags are reset to $0$. The above procedure is equivalent to the standard binary semaphores algorithm to achieve process synchronization ~\cite{silberschatz2018operating}. 

\begin{example}
In Figure \ref{fig:eg6}, the agents are following the modified trajectories in Figure \ref{fig:eg7}c. The priority of the agent changes once the agent reaches its final state. More precisely, the agent gets the lowest priority once it reaches its final state. 
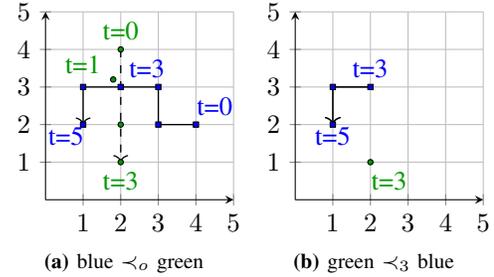
\begin{figure}[htb!]
\centering
\definecolor{green1}{rgb}{0,0.6,0}
\subfloat[blue $\prec_o$ green]{
\begin{tikzpicture}[line join=round,x=1cm,y=1cm]
\begin{axis}[
scale=0.5,
x=1cm,y=1cm,
axis lines=middle,
ymajorgrids=true,
xmajorgrids=true,
xmin=0,
xmax=5,
ymin=0,
ymax=5,
xtick={-3,-2,...,35},
ytick={-7,-6,...,17},]
\draw [line width=0.5pt] (4,2)-- (3,2);
\draw [line width=0.5pt] (3,2)-- (3,3);
\draw [line width=0.5pt] (3,3)-- (2,3);
\draw [line width=0.5pt] (2,3)-- (1,3);
\draw [line width=0.5pt,->] (1,3) -- (1,2);
\draw [->,line width=0.5pt,dashed] (2,4) -- (2,1);
\begin{scriptsize}
\draw [fill=green1] (2,4) circle (1pt);
\draw[color=green1] (2,4.5) node {t=0};

\draw [fill=green1] (1.8,3.2) circle (1pt);
\draw[color=green1] (1,3.6) node {t=1};

\draw [fill=green1] (2,2) circle (1pt);
%\draw[color=green1] (2,1.5) node {t=2};

\draw [fill=green1] (2,1) circle (1pt);
\draw[color=green1] (2,0.5) node {t=3};

\draw [fill=blue] (4,2) \Square{1pt};
\draw[color=blue] (4.5,2.5) node {t=0};

\draw [fill=blue] (3,2) \Square{1pt};
%\draw[color=blue] (3.5,1.5) node {t=1};

\draw [fill=blue] (3,3) \Square{1pt};
%\draw[color=blue] (3,3.5) node {t=2};

\draw [fill=blue] (2,3) \Square{1pt};
\draw[color=blue] (2.7,3.5) node {t=3};

\draw [fill=blue] (1,3) \Square{1pt};
%\draw[color=blue] (0.6,3.3) node {t=4};

\draw [fill=blue] (1,2) \Square{1pt};
\draw[color=blue] (0.55,1.7) node {t=5};
\end{scriptsize}
\end{axis}
\end{tikzpicture}}
\subfloat[green $\prec_3$ blue]{
\begin{tikzpicture}[line join=round,x=1cm,y=1cm]
\definecolor{green1}{rgb}{0,0.6,0}
\begin{axis}[
scale=0.5,
x=1cm,y=1cm,
axis lines=middle,
ymajorgrids=true,
xmajorgrids=true,
xmin=0,
xmax=5,
ymin=0,
ymax=5,
xtick={-1,0,...,25},
ytick={-1,0,...,14},]
\draw [line width=0.6pt] (2,3) -- (1,3);
\draw [->,line width=0.6pt] (1,3) -- (1,2);
\begin{scriptsize}
\draw [fill=green1] (2,1) circle (1pt);
\draw[color=green1] (2.5,0.5) node {t=3};

\draw [fill=blue] (1,3) \Square{1pt};

\draw[color=blue] (2,3.5) node {t=3};
\draw [fill=blue] (2,3) \Square{1pt};

\draw [fill=blue] (1,2) \Square{1pt};
\draw[color=blue] (1,1.7) node {t=5};
\end{scriptsize}
\end{axis}
\end{tikzpicture}}
\caption{ The blue and green agents are following their modified trajectories from Example 3. Initially, the blue agent has a lower priority; hence, it is forced to modify its path. When the green agent reaches its final state at $t=3$, the green agent is assigned a lower priority. The blue agent's priority is higher than the green agent's. Again at $t=5$, the priorities change since the blue agent has reached its final state.}
\label{fig:eg6}
\end{figure}
\end{example}

\paragraph*{Extension to Arbitrary Number of Agents.}

We extend the procedure outlined above to multiple agents. Each pair of agents $u_i$ and $u_j$ maintain two Booleans between them (each of them is analogous to a binary semaphore) that are used to measure relative progress. 
%\Dj{That is,} each agent maintains a vector of \Dj{Booleans} corresponding to all the other agents in the system. 
%\Dj{Every agent} also maintains a set for \Dj{keeping a record} of all the other agents that is has communicated with so far.
Formally, $\overline{C_u} = (c_u^1,c_u^2,\ldots ,c_u^{|\Agents|})$ is a vector of Boolean flags for maintaining progress of $u$ with respect to the other agents and $B_u$ is a set maintained by $u$ for tracking the agents it has communicated with during the current final state.
Initially, $B_u = \emptyset$ and $\overline{C_u} = \overline{0}$. The flag $c_u^v$ on agent $u$ records the progress of $u$ with respect to $v$. If $c^v_u$ is 1, agent $u$ has recorded that it has \emph{finished a goal} ($goal_{u,t} = \top$) after communicating with $v$ $(v \in B_u)$. Whenever the corresponding progress measures $c_u^v$ and $c_v^u$ are equal and the agents are in the same communication group, the Boolean flags are reset to $0$. The exact algorithm is presented in Figure \ref{fig:flow}b.

Let $c^v_u(t)$ denote the value of the flag $c^v_u$ at time $t$.
We use $a \prec^t b$ to denote that $c_a^b(t) =1,~c_b^a(t)=0$ and $a=^tb$ to denote $c_a^b(t) = c_b^a(t)$. Observe that for any pair of agents $a$ and $b$ either $a =^t b$ or $a \prec^t b$ or $b \prec^t a$. %\Dj{Moreover,} $\prec^t$ as defined is a \emph{partial}-order.

%\begin{remark}
%If none of the agents have completed any goal, then $a =^t b =^t c$. \end{remark}

\begin{prop}
%Let $a, b, c$ be agents. 
If $a \prec^t b$ and $b \prec^t c$, then no agent $d$ exists such that $c \prec^t d$ and $d \prec^t a$.
\end{prop}
\begin{proof}[proof by contradiction]
For any agent $u$, let $comp(u)$ denote the earliest time $t'$ such that $goal_{u,t'} = \top$.
If $a \prec^t b$, then $c^b_a = 1$ and $c_b^a = 0$, i.e., agent $a$ has reached its final state, but agent $b$ has not.  Similarly, $b \prec^t c$ implies that agent $b$ has reached its final state, but agent $c$ has not. Therefore, 
\begin{equation}
comp(a) < comp(b) < comp(c).
\end{equation}
%$comp(a) < comp(b) < comp(c)$.
%\Dj{We denote the order of completion between agents $a, b$ and $c$ as} $c \cdots b \cdots a$. %Thus, $a \prec^t b$ and $b \prec^t c$ implies $a \prec^t c$.
Now, suppose there exists an agent $d$ such that $c \prec^t d$ and $d \prec^t a$, then by the same argument, the order of last completed goals among $a,c$ and $d$ is 
\begin{equation}
    comp(c) < comp(d) < comp(a).
\end{equation}
%$comp(a) < comp(d) < comp(c)$. 
(1) contradicts (2). Therefore, there cannot exist agent $d$ such that $c \prec^t d$ and $d \prec^t a$. %\qed 
\end{proof}

\begin{corollary}
$\prec^t$ is a partial-order.
\end{corollary}

\begin{prop}
There exists a total order $\prec_t$ that respects $\prec^t$.
\end{prop}
\begin{proof}
Define $\prec_t$ as 
%For any pair \Dj{of} agents $a$ and $b$, define $\prec_t$ as
\begin{align*}
    a \prec_t b &\text{ if }
                    i)~a \prec^t b \text{ or }
                    ii)~a =^t b \text{ and } a \prec_0 b, \\
    b \prec_t a &\text{ otherwise,}                
\end{align*} where $a$ and $b$ are some agents.
$\prec_t$ as defined is a total order. %completing the proof. 
%\qed
\end{proof}

\noindent Henceforth, we use $\prec_{0}$ to generate $\prec_t$ that respects $\prec^t$. %$\prec^t$ as defined is transitive, i.e., $a \prec^t b$ and $b \prec^t c$, then $a \prec^t c$.
%This is shown in Algorithm in Figure \ref{fig:flow} (right).

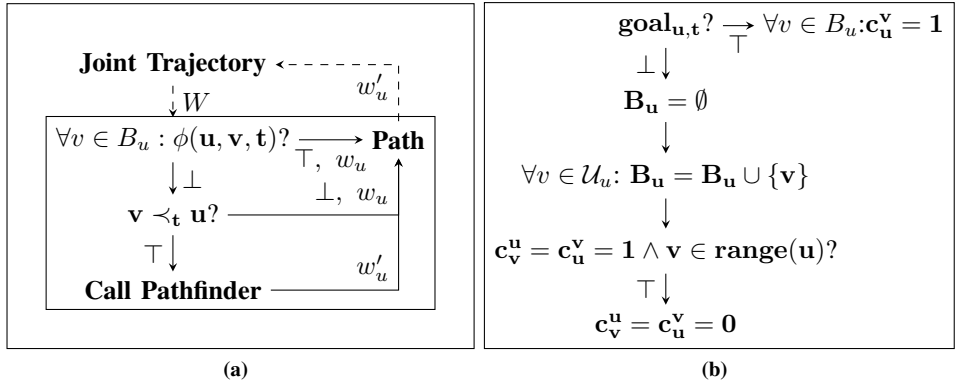
\begin{figure*}[!htb]
\centering
%\begin{subfigure}[t]{0.5\textwidth}
\tikzstyle{line} = [draw]
\subfloat[]{
\begin{tikzpicture}[%
    >=stealth,
    node distance=1cm and 3cm,
    on grid,
    auto,
    >=stealth,
]
\node (safe) {$\forall v \in B_u: \mathbf{\phi (u,v,t)}$?};
\node[below=of safe] (no) {$\mathbf{v \prec_t u}$?};
\node[right = of safe] (yes) {\textbf{Path}};
\node[below = of no] (pathfinder) {\textbf{Call Pathfinder}};
\path [line,->] (safe) -- node [anchor = west] {$\bot$} (no);
\path [line,->] (safe) -- node [anchor = north] {$\top,~w_u$} (yes);
\path [line,->] (no) -- node [anchor = east] {$\top$} (pathfinder);
\path [line,->] (no) -| node [anchor = south east] {$\bot,~w_u$} (yes);
\path [line,->] (pathfinder) -| node [anchor = south east] {$w_u'$} (yes);
\node[above = of safe] (jt) {\textbf{Joint Trajectory}};
\node[fit= (jt) (safe) (no) (yes) (pathfinder), draw, inner sep=5.3mm] (total) {};
\node[fit= (safe) (no) (yes) (pathfinder), draw, inner sep=0.1mm] (decision) {};
\node[above = of safe] (for) {};
\path [line,dashed, ->] (yes) |- node [anchor = north east]{$w_u'$} (jt);
\path [line, dashed, ->] (jt) -- node {$W$} (safe);
\end{tikzpicture}} 
%\end{subfigure}
%\begin{subfigure}[t]{0.3\textwidth}
%\tikzstyle{line} = [draw]
\subfloat[]{
\begin{tikzpicture}[%
    >=stealth,
    node distance=1cm and 2.5cm,
    on grid,
    auto,
    >=stealth
]
\node (goal) {$\mathbf{goal_{u,t}}$?};
\node [right = of goal] (increase_progress)  {$\forall v \in B_u$:$\mathbf{c_u^v = 1}$};
\node [below = of goal] (resetb) {$\mathbf{B_u = \emptyset}$};
\node [below = of resetb] (add) {$\forall v \in \mathcal{U}_u$: $\mathbf{B_u = B_u \cup \{v\} }$};
\node [below = of add] (can) {$\mathbf{c_v^u = c_u^v = 1} \land \mathbf{v \in range(u)}$?};
\node [below = of can] (rc) {$\mathbf{c_v^u = c_u^v = 0}$};
\node[fit=(goal) (increase_progress) (resetb) (add) (can) (rc), draw, inner sep=0.1mm] (decision) {};
\path [line, ->] (goal) -- node [anchor = north] {$\top$} (increase_progress);
\path [line, ->] (goal) -- node [anchor =  east] {$\bot$} (resetb);
\path [line, ->] (resetb) -- (add);
\path [line, ->] (add) -- (can);
\path [line, ->] (can) -- node [anchor = east] {$\top$} (rc);
\end{tikzpicture}}
%\end{subfigure}
\caption{(a) Algorithm for the enforcer $S(u,t)$ to decide if the pathfinder should be called at time $t$. (b) Algorithm to maintain the priorities.}
\label{fig:flow}
\end{figure*}

\subsubsection*{Decentralized Enforcement.}

So far, we have described the components of the enforcers onboard an agent. In traditional shield synthesis, $S_1$ is a function that is fully constructed and used as the shield \cite{BloemKKW15}. In contrast, here $S_1$ is a partial function which is when required.  Figure \ref{fig:flow}a presents the algorithm to determine the calls to the pathfinder and Figure \ref{fig:flow}b presents the ordering mechanism to update the priorities by modifying the corresponding flags. 

For some agent $u_i$, if $w'_{u_i}$ is the path returned by the pathfinder, then $S_1(u,t)(O,V,W)=(O,V,\Replace_i(W,(v_{u_i},w'_{u_i})))$. That is, the path for the agent $u_i$ has been replaced with $w'_{u_i}$, with the paths of the other agents unaffected and their priorities unchanged.
If the pathfinder is never called, then the path does not get modified, i.e., $S_1(u_i,t)(O,V,W) = (O,V,W)$.

The following lemma is a direct consequence of the construction of the pathfinder graph $G_{u}^t$ and a path between $v_{init}$ and $F$. 
\begin{lemma}
For all $t' \in [t,t+\ell]$ and $u \in \Agents$, if agent $u$ moves along the trajectory returned by $S_1(u,t')$, then $\Safety(t') = \top$.
\label{lemma:correct}
\end{lemma}
 
\noindent Next, we prove that the agent with the highest priority is able to progress without any deviation.
\begin{lemma}
If agent $u$ has the highest priority according to $\prec_t$ then it will reach its final state without any modifications. Moreover, if a $\ell$--stabilizing centralized shield can ensure safety, then the enforcers also can ensure safety.
%in at most $\ell+k$ steps. 
\label{lemma:hp_lt}
\end{lemma}
\begin{proof}
If $u$ has the highest priority by $\prec_t$, then for all $v \in \Agents$, it is either the case that $c_u^v=0$ and $c_v^u=1$ (or) $c_u^v = c_v^u$ and $v <_{0} u$.
In either case, $v$ finds a new path if a safety violation is detected. Since a centralized shield can ensure safety, it implies that there is at least one safe position for $v$ in the occupancy graph $O^t_v$. Assumption 1 states that $G$ is 1-edge connected therefore, there is a safe vertex such that the path length is at most $\ell$. By the pathfinder algorithm, the trajectory of $v$ is modified. Similarly, all other agents modify their trajectories in the case of a safety violation. 
%Therefore, by assumption 1, agent $u$ can reach its goal in $2\ell + k$ steps. %\qed
\end{proof}

\begin{corollary}
In the worst case, the distance between $u$ and its final state maybe $|\Agents|(\ell)$.
\label{lmt}
\end{corollary}
We now prove that the other agents are also guaranteed to make progress. The following theorem bounds the maximum deviation from the original trajectory.

\begin{lemma}[Main]
Enforcer on agent $u$ may cause a deviation from the intended trajectory for at most $|\Agents|^2(\ell)$ steps before the final state is reached. 
\label{lemma:main}
\end{lemma}
\begin{proof}
%Consider the case where agents $a$ and $b$ have a safety violation and agent $a$ consequently deviates from its trajectory. In the future, if $a$ and $b$ have a safety violation, then agent $b$ deviates before agent $a$.
%We consider the case when $u$ is the lowest priority agent and it has to use the occupancy graph to find a path to ensure safety until gets the highest priority. In the worst case when $u$ obtains the highest priority, the distance between $u$ and its final state may be $|\Agents|(\ell)$. Now by lemma \ref{lemma:hp_lt}, it can proceed without any modifications. T
In the worst case, any agent $u$ may be forced to use the occupancy graph during every call to the pathfinder, before agent $u$ has the highest priority according to $\prec_t$. In worst case, agent $u$ might be the lowest priority vertex to start with. Hence, agent $u$ may require $|\Agents|-1(|\Agents|\ell)$ steps to get the highest priority. At this stage by Corollary \ref{lmt}, agent $u$ requires $|\Agents|(\ell)$ steps to reach its goal. 
 %\qed
\end{proof}

The next theorem establishes that the enforcers we synthesize satisfy the properties stated in Problem 1.

\begin{theorem}[Main]
The set $\{S(u,t)|u \in \Agents \}$ of enforcers are i) correct, ii) deviate minimally, iii) bounded, and iv) complete.
\end{theorem}
\begin{proof}
As a consequence of Lemma \ref{lemma:main}, the maximum number of steps that any agent needs to reach its final state is bounded by $(|\Agents|^2\ell)$. Therefore, the synthesized enforcers are all $(\ell,|\Agents|^2\ell)$--enforcers and \emph{bounded}. If no safety violation is detected then the pathfinder is never called. Hence, the enforcers also \emph{deviate minimally}. Moreover, Lemma \ref{lemma:correct} establishes that the enforcers are \emph{correct}. Further by Lemma \ref{lemma:hp_lt} the enforcer can ensure safety, if a $\ell$--stabilizing centralized sheild can ensure safety. That is, it is \emph{complete}.
%in the interval $[t,t+\ell]$. 
%\qed
\end{proof}

%\subsubsection*{Complexity.}
\noindent The main complexity result of the paper, where we bound the worst-case synthesis time, is formalized in the theorem below.
% \Dj{Given look-ahead $\ell$ and maximum deviation length $k$ from a single call, the} pathfinder constructs a graph for agent $u$ of size at most $(k+\ell) |\Agents|$. If $k$ and $\ell$ are fixed, then the size of the graph constructed by the \Dj{pathfinder} is $\text{LINEAR}(|\Agents|)$. \Dj{Moreover,} the number of edges in this graph is \Dj{also} at most $(k+\ell)|\Agents|$. The time complexity of solving a search in this graph is $\mathcal{O}((k+\ell)|\Agents|)$. \Dj{In the worst case,} all the agents are in the same communication group and the lowest priority agent may have to modify its trajectory at must $|\Agents| - 1$ times. \Dj{Hence,} the enforcer on an agent can take up to $\mathcal{O}(|\Agents|^2)$ time to modify its trajectory. 
\begin{theorem}
Given fixed look-ahead $\ell$ and maximum deviation length $k$, the enforcer on an agent takes $\mathcal{O}(|\Agents|^2)$ time to modify the corresponding agent's trajectory. 
\end{theorem}
\begin{proof}
The pathfinder constructs a graph for agent $u$ of size at most $(k+\ell) |\Agents|$. If $k$ and $\ell$ are fixed, then the size of the graph constructed by the pathfinder is $\text{LINEAR}(|\Agents|)$. Moreover, the number of edges in this graph is also at most $(k+\ell)|\Agents|$. The time complexity of solving a search in this graph is $\mathcal{O}((k+\ell)|\Agents|)$. In the worst case, all agents are in the same communication group and the lowest priority agent may have to modify its trajectory at most $|\Agents|$ times. %\qed
\end{proof}
\section{Extension to Trajectories of Arbitrary Length}

In problem 1, we assume that the length of any trajectory is bound by some constant $\ell$. %If this $\ell$ is large, then the resultant graph constructed by the pathfinder also becomes large.
In this section, we define the problem for trajectories of arbitrary length and propose a solution. We show how to use the ordering mechanism to ensure that the enforcers can ensure safe behavior even when the trajectories are of arbitrary length.

\begin{prob}
Given a set $\Agents = \{u_1,\dots,u_n\}$ of agents and a set $\{(v_u,w_u)~|~u \in \Agents \}$ of their trajectories, 
construct a set $\{S(u,t)~|~u \in \Agents \text{ and } t \in \Time\}$ of enforcers such that these enforcers are correct, cause minimum deviation and bounded. 
\end{prob}

Luckily, we do not have to change the entire synthesis procedure to solve this problem. We artificially restrict the input to the enforcers. $w_u(0)$ be the trajectory (possibly infinite) for agent $u$. $v_u$ be the start state of the agent. We divide $w_u$ into blocks of length $\ell$, $w_u(0) = w_u[0:\ell-1]\cdot w_u[\ell:2\ell-1]\dots$. The enforcer $S(u,0)$ uses $w_u(0)$ as its trajectory, once it reaches the final state $\hat{\delta}(v_u,w_u(0))$, it uses $\bigg( (\hat{\delta}(v_u,w_u(0)),  w_u[\ell:2\ell-1] \bigg)$ as the new trajectory for the enforcer synthesis. Algorithm \ref{Trajectory_update} describes this procedure.

%This procedure is presented in Algorithm \ref{Trajectory_update}.

In the algorithm, whenever the agent reaches its final state, its trajectory is updated with the next $\ell$ moves and its flags are also suitably reset. The correctness of this procedure is a direct consequence of the main theorem.

\begin{smaller}
\begin{algorithm}
\KwResult{Trajectory of length at most length $\ell+k$}
$i=0$\;
Initialize $w_u$ and $v_u$\;
$v_u^f = \hat{\delta}(v_u,w_u(0)[0:\ell])$\;
$w_u(0) = w_u[0:\ell]$\;
\While{True}{
Update Communication Groups \;
Increment Time \;
\If{$goal_{u,t} = \top$}{
$i++$\;
$v' = v_u^f$\;
$w_u(t) = w_u[i\ell:(i+1)\ell]$\;
$v_u^f = \hat{\delta}(v_u,w_u[i\ell:(i+1)\ell-1])$\;
$v_u = v'$\;
}
\ForAll{$v$ in $\Agents_u$}{
\If{($u$, $v$ violate safety) $\land (v =_{t} u)$}
{
    \If{$v \prec_{0} u$}{Call the pathfinder on $u$ and find a new path \;} 
}
\ElseIf{($u$, $v$ violate safety) $\land (v \prec_{t} u)$}
{
    Call the pathfinder on $u$ and find a new path \;
    $B_u = \emptyset$ \;
}
}
t++ \;
}
\caption{Enforcer on agent $u$}
\label{Trajectory_update}
\end{algorithm}
\end{smaller}

\section{Experimental Evaluation}
We evaluate the performance of the runtime decentralized enforcer synthesis framework in the context of collision-avoidance for multi-agent systems. Specifically, we use the collision-avoidance safety function defined in Example 3 and the pathfinder construction described in Example 7. 
 The implementation of the system for ensuring safety from collisions uses the general pseudocode presented in Algorithm \ref{Trajectory_update}.

\subsubsection*{Comparison with Centralized Shields.}

We compare the modified trajectories of two agents equipped with decentralized enforcers that are synthesized at runtime with two other agents whose behaviors are modified by a centralized shield synthesized at design-time using the algorithm from \cite{multiagentshield} in a 5x5 grid world. The intended trajectories of the agents in both scenarios are the same. 

%The shields modify the trajectories to ensure there are no collisions between two agents. Figure \ref{fig:traditional} shows the behavior of the agents with the centralized shields acting on them for the scenario presented in Figure  \ref{fig:eg7}b. In this scenario, the agents have no look-ahead, i.e., $\ell = 1$. We show in Table~\ref{tab:fullresults} that the resulting state space from incorporating look-ahead $\ell > 1$ is so large that the design-time synthesis problem becomes computationally intractable even for two agents in a 5x5 grid. The green and blue agents detect at $t=1$ that there will be a collision at $t=2$ if they follow their original trajectories. The trajectory of the blue agent is modified and it reaches (2,2) at $t=4$ instead of reaching (2,2) at $t=2$. The shield makes the blue agent converge to (2,2) as soon as possible.

%In contrast, 
Decentralized enforcers can incorporate look-ahead $\ell \geq 1$. In the case with $\ell=1$, the decentralized enforcers behave precisely the same as the centralized enforcer as they can only detect collisions in the next step. In the case with look-ahead $\ell=3$, i.e., with further look-ahead,  we recover the solution presented in Figure~\ref{fig:eg6}. 
In the case with $\ell=2$, the enforcers induce a different behavior. At  $t=0$, only the collision is detected, but the final state for the blue agent is (2,2) instead of (1,2) as in the previous case. The intended trajectory is updated when the agents have reached their current goals (in this case, this update happens at (2,2) for both the agents). The effect of the enforcer is shown in Figure \ref{fig:ds_test}.

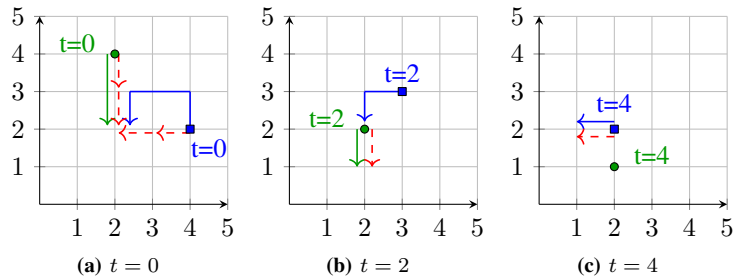
\begin{figure*}[t!]
\centering
\definecolor{green1}{rgb}{0,0.6,0}
\subfloat[$t=0$]{
\begin{tikzpicture}[line join=round,x=1cm,y=1cm]
\begin{axis}[
scale=0.5,
x=1cm,y=1cm,
axis lines=middle,
ymajorgrids=true,
xmajorgrids=true,
xmin=0,
xmax=5,
ymin=0,
ymax=5,
xtick={-3,-2,...,35},
ytick={-7,-6,...,17},]
\draw [line width=0.6pt,color=blue] (4,2) -- (4,3);
\draw [line width=0.6pt,color=blue] (4,3) -- (2.4,3);
\draw [line width=0.6pt,color=blue,->] (2.4,3) -- (2.4,2.1);

\draw [->,line width=0.6pt,color=red,dashed] (4,1.9) -- (3.1,1.9);
\draw [->,line width=0.6pt,color=red,dashed] (3,1.9) -- (2.1,1.9);

\draw [->,line width=0.6pt,color=green1] (1.8,4) -- (1.8,2.1);
\draw [->,line width=0.6pt,color=red,dashed] (2.1,4) -- (2.1,3.1);
\draw [->,line width=0.6pt,color=red,dashed] (2.1,3) -- (2.1,2.1);
\begin{scriptsize}
\draw [fill=green1] (2,4) circle (1.5pt);
\draw[color=green1] (1,4.3) node {t=0};
\draw [fill=blue] (4,2) \Square{1.5pt};
\draw[color=blue] (4.5,1.5) node {t=0};
\end{scriptsize}
\end{axis}
\end{tikzpicture}}
\subfloat[$t=2$]{
\begin{tikzpicture}[line join=round,x=1cm,y=1cm]
\begin{axis}[
scale=0.5,
x=1cm,y=1cm,
axis lines=middle,
ymajorgrids=true,
xmajorgrids=true,
xmin=0,
xmax=5,
ymin=0,
ymax=5,
xtick={-3,-2,...,35},
ytick={-7,-6,...,17},]
\draw [line width=0.6pt,color=blue] (3,3)-- (2,3);
\draw [line width=0.6pt,color=blue,->] (2,3)-- (2,2.2);

\draw [->,line width=0.6pt,color=green1] (1.8,2) -- (1.8,1);
\draw [->,line width=0.6pt,color=red,dashed] (2.2,2) -- (2.2,1);
\begin{scriptsize}

\draw [fill=green1] (2,2) circle (1.5pt);
\draw[color=green1] (1,2.3) node {t=2};

\draw [fill=blue] (3,3) \Square{1.5pt};
\draw[color=blue] (3,3.5) node {t=2};

\end{scriptsize}
\end{axis}
\end{tikzpicture}}
\subfloat[$t=4$]{
\begin{tikzpicture}[line join=round,x=1cm,y=1cm]
\begin{axis}[
scale=0.5,
x=1cm,y=1cm,
axis lines=middle,
ymajorgrids=true,
xmajorgrids=true,
xmin=0,
xmax=5,
ymin=0,
ymax=5,
xtick={-3,-2,...,35},
ytick={-7,-6,...,17},]

\draw [line width=0.6pt,color=blue,->] (2,2.2)-- (1,2.2);
\draw [line width=0.6pt,color=red,->,dashed] (2,1.8)-- (1,1.8);
\begin{scriptsize}

\draw [fill=green1] (2,1) circle (1.5pt);
\draw[color=green1] (3,1.3) node {t=4};

\draw [fill=blue] (2,2) \Square{1.5pt};
\draw[color=blue] (2,2.7) node {t=4};

\end{scriptsize}
\end{axis}
\end{tikzpicture}}
\caption{The trajectories of the blue and green agents have been modified by the respective enforcers acting on them to ensure no collisions. The first goal for the two agents is (2,2), as their look-ahead $\ell' = 2$. Once they reach (2,2), their intended trajectories are updated, which is shown in (b) and (c).}
\label{fig:ds_test}
\end{figure*}

As shown in these examples, the look-ahead parameter $\ell$ impacts the modified behavior. The agent has an increased ability to prevent future collisions with a larger value of $\ell$. This enhanced ability to prevent collisions comes at the cost of the synthesis time as the size of the graph constructed by the pathfinder increases. But, it does not affect the maximum length of the deviation.

\subsubsection*{Scalability.}

\begin{table*}[!t]
\centering
\begin{tabular}{cccccccccccc} 
\\
\toprule
 $|\mathcal{U}|$                 & States  & $~\ell~$  & $~k~~$  & \multicolumn{2}{c}{\begin{tabular}[c]{@{}c@{}}Centralized \\ game graph \end{tabular}} & \multicolumn{2}{c}{\begin{tabular}[c]{@{}c@{}}Decentralized \\ pathfinder graph \end{tabular}} & $|\Lang_I|$  & $|\Lang_O|$  & \multicolumn{2}{l}{\begin{tabular}[c]{@{}l@{}}Decentralized\\synthesis time \end{tabular}}  \\
                                 &         &           &         & $|V|$       & $|E|$                                                                    & $|V|$  & $|E|$                                                                                 &              &              & (best) & (worst)                                                                                              \\ 
\midrule
3                                & $3^2$   & 3         & 3       & $10^8$      & $10^{12}$                                                                & 18     & 18                                                                                    & $10^{2}$     & $10^{3}$     & 0.089  & 0.267                                                                                                \\
\rowcolor[rgb]{0.925,0.957,1} 3  & $3^2$   & 5         & 5       & $10^{11}$   & $10^{18}$                                                                & 30     & 30                                                                                    & $10^{3}$     & $10^{7}$     & 0.092  & 0.276                                                                                                \\
3                                & $3^2$   & 10        & 5       & $10^{20}$   & $10^{30}$                                                                & 45     & 45                                                                                    & $10^{7}$     & $10^{10}$    & 0.093  & 0.279                                                                                                \\
\rowcolor[rgb]{0.925,0.957,1} 3  & $5^2$   & 5         & 5       & $10^{13}$   & $10^{19}$                                                                & 30     & 30                                                                                    & $10^{5}$     & $10^{7}$     & 0.092  & 0.276                                                                                                \\
3                                & $10^2$  & 5         & 5       & $10^{15}$   & $10^{21}$                                                                & 30     & 30                                                                                    & $10^{6}$     & $10^{7}$     & 0.092  & 0.276                                                                                                \\
\rowcolor[rgb]{0.925,0.957,1} 3  & $50^2$  & 5         & 5       & $10^{19}$   & $10^{25}$                                                                & 30     & 30                                                                                    & $10^{7}$     & $10^{7}$     & 0.092  & 0.276                                                                                                \\
5                                & $3^2$   & 5         & 5       & $10^{19}$   & $10^{25}$                                                                & 50     & 50                                                                                    & $10^{3}$     & $10^{7}$     & 0.2    & 1                                                                                                    \\
\rowcolor[rgb]{0.925,0.957,1} 5  & $5^2$   & 5         & 5       & $10^{22}$   & $10^{28}$                                                                & 50     & 50                                                                                    & $10^{5}$     & $10^{7}$     & 0.2    & 1                                                                                                    \\
5                                & $10^2$  & 5         & 5       & $10^{25}$   & $10^{31}$                                                                & 50     & 50                                                                                    & $10^{6}$     & $10^{7}$     & 0.2    & 1                                                                                                    \\
\rowcolor[rgb]{0.925,0.957,1} 20 & $50^2$  & 3         & 3       & $10^{104}$  & $10^{107}$                                                               & 120    & 120                                                                                   & $10^{6}$     & $10^{6}$     & 0.41   & 8.2                                                                                                  \\
20                               & $50^2$  & 5         & 5       & $10^{128}$  & $10^{134}$                                                               & 200    & 200                                                                                   & $10^{7}$     & $10^{7}$     & 0.42   & 8.4                                                                                                  \\
\rowcolor[rgb]{0.925,0.957,1} 20 & $50^2$  & 10        & 5       & $10^{188}$  & $10^{197}$                                                               & 300    & 300                                                                                   & $10^{10}$    & $10^{10}$    & 1.29   & 25.8                                                                                                 \\
30                               & $50^2$  & 10        & 5       & $10^{282}$  & $10^{291}$                                                               & 450    & 450                                                                                   & $10^{10}$    & $10^{10}$    & 1.62   & 48.6                                                                                                 \\
\rowcolor[rgb]{0.925,0.957,1} 40 & $50^2$  & 10        & 5       & $-$         & $-$                                                                      & 600    & 600                                                                                   & $10^{10}$    & $10^{10}$    & 1.64   & 65.6                                                                                                 \\
50                               & $50^2$  & 10        & 5       & $-$         & $-$                                                                      & 750    & 750                                                                                   & $10^{10}$    & $10^{10}$    & 1.67   & 83.5                                                                                                 \\
\rowcolor[rgb]{0.925,0.957,1} 60 & $50^2$  & 10        & 5       & $-$         & $-$                                                                      & 900    & 900                                                                                   & $10^{10}$    & $10^{10}$    & 1.7    & 102 \\
\bottomrule \\
\end{tabular}
\caption{Comparison of state space sizes between centralized and decentralized online enforcement approaches with reported synthesis times for the decentralized approach. As the enforcers are only synthesized as needed for the relevant agents, we report both the worst and best-case total synthesis times(sec) for all relevant enforcers for every detected collision. In case of the number of vertices and edges in the centralized approach the order of magnitude are shown. $\Lang_I$ and $\Lang_O$ are the input alphabet and the output alphabet respectively.}
\label{tab:fullresults}
\end{table*}

We built a multi-agent system where the agents are equipped with the decentralized enforcer framework for collision-avoidance in Python. The distributed nature of the system is modeled using shared memory. The size of the grid world, the look-ahead length $\ell$, the communication constant $d$, and the length $k$ of the maximum deviation by one use of the pathfinder are user inputs. The original trajectories for the agents are random. We record the effect of $\ell$ and $k$ on the synthesis time for modified behavior. 
The results are obtained on an Intel Core-i7 CPU @ 2.2 GHz with 16GB of RAM. We set $d = \ell$ in all the experiments. We show the results of these experiments in Table \ref{tab:fullresults}. 

To synthesize the centralized shield, a safety game is solved. We show the size of the game graph (in the order of magnitude) for the different scenarios.  The large size explains why the design-time synthesis of centralized shields is infeasible in multi-agent settings.
For comparison, we also show the exact size of the graph constructed by the pathfinder for each scenario. Finally, we record the best and the worst-case synthesis times in the decentralized setting. Observe that the synthesis times are the same if $\ell,k$, and $|\Agents|$ are the same. 
Table \ref{tab:fullresults} shows that the synthesis time does not depend on the number of states in the environment. This observation is consistent with the earlier analysis. The worst case is when all the agents interfere with one another and they have global information. In this scenario, only the agent with the highest priority can progress without modifications. Every other agent has to wait for the agents with a lower priority to fix their trajectories. Nevertheless, Table \ref{tab:fullresults} shows that that the synthesis time is in the order of a few seconds.

\subsubsection*{ROS/Gazebo Simulation.}

In this section, we demonstrate the implementation of the enforcers in a high-fidelity simulation environment shown in Figure~\ref{fig:sims}. We \emph{do not} use any in-built obstacle avoidance library for low level collision avoidance. We tested two different scenarios in this environment. In the first scenario, we used a new safety function to prevent collision with obstacles (buildings) and ensure that the Manhattan distance between any two agents is at least 2 units. We randomly generated trajectories for the agents. A simulation with eight agents is presented in the video \footnote{ Video can be found at https://tinyurl.com/yhdddpm6}. 

In the second setting shown in~\ref{fig:sims} (top) nine agents are occupying locations corresponding to a 3x3 square. The agent in the center has to escape the confinement. Further, no trajectories for the other agents are given. In this scenario, the central agent chooses its trajectory when it gets the highest priority. The agents along the path of this agent have to make way. During this process, sometimes a set of agents may actually have to make way for one another. This scenario is presented in the video \footnote{Video can be found at https://tinyurl.com/ygwaplcp}.

\begin{figure}[!htp]

\centering
\includegraphics[width=.5\textwidth]{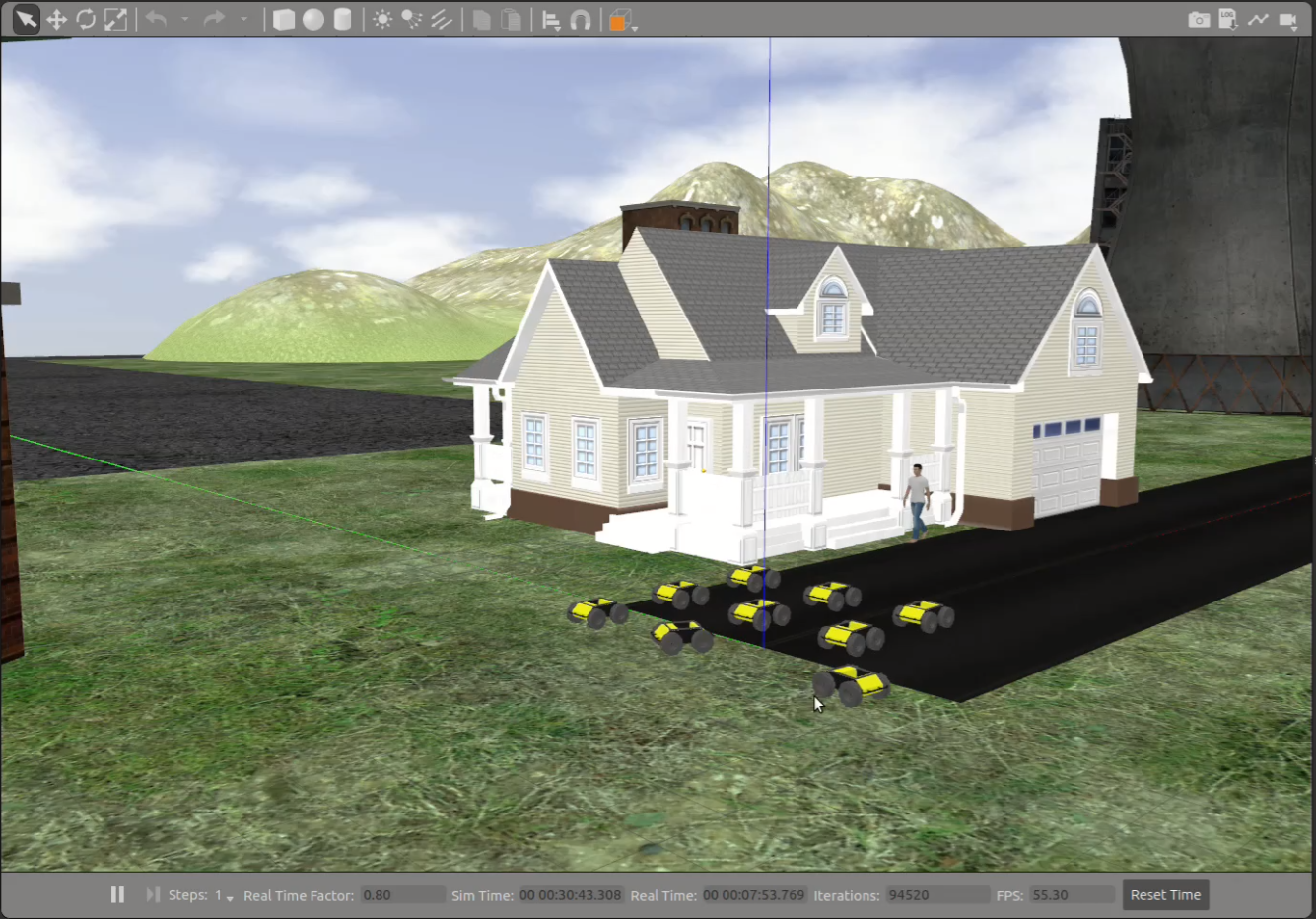}

\caption{Custom Gazebo environment used for the experiments.}
\label{fig:sims}

\end{figure}

\section{Conclusion}

We present an online synthesis approach for runtime enforcers that guarantee local safety in multi-agent systems. Moreover, this approach is decentralized, since new behavior is synthesized onboard each agent when necessary. The algorithm we present does not require global information on the states of the other agents in the system. It only requires the information about the agents in the same communication group. With minor assumptions, we prove the correctness of this approach in enforcing safety and also prove that all the agents progress per their original plan by proving a bound on the maximum deviation. Additionally, we also provide a condition which guarantees completeness. More specifically, we show that if a $\ell$-stabilizing centralized shield can guarantee correctness, then $(\ell,|\Agents|^2 \ell)$--enforcers can also guarantee correctness.  We further prove that this synthesis scales with the number of agents. In the future, we plan to consider the enforcement of general safety properties that do not require the \emph{local} restriction. Additionally, we will attempt to solve the problem in more realistic communication architectures. 

\nocite{kok2003multi}
\nocite{wongpiromsarn2011formal}
\nocite{sujit2005multi}
\nocite{ong2007multi}
\nocite{falconeruntime1}
%\nocite{KonighoferABHKT17}
\nocite{rasmussen2018brief}
\nocite{guestrin2002multiagent}
\nocite{Pnueli77}
\nocite{LecturesOnRuntimeVerification}
\nocite{BauerLS11}
\nocite{BauerF16}

\bibliographystyle{IEEEtran}
\bibliography{ref}

\end{document}